\documentclass[11pt]{article}
\usepackage[T1]{fontenc}
\usepackage{lmodern}
\usepackage[margin=1in]{geometry}
\usepackage{amsmath,amssymb,amsthm}
\usepackage{microtype}
\usepackage{algorithm,algpseudocode}
\usepackage[hidelinks]{hyperref}
\usepackage{xcolor}

\newtheorem{theorem}{Theorem}[section]
\newtheorem{lemma}[theorem]{Lemma}
\newtheorem{corollary}[theorem]{Corollary}
\newtheorem{definition}[theorem]{Definition}

\newcommand{\poly}{\operatorname{poly}}
\newcommand{\MaxSAT}{\operatorname{MaxSAT}}
\newcommand{\TV}{\operatorname{TV}}
\newcommand{\BIS}{\#\mathrm{BIS}}
\newtheorem{question}{Question}

\numberwithin{equation}{section}
\allowdisplaybreaks[1]

\title{Separations between $\#\mathrm{BIS}$ and $\#\mathrm{IS}$}
\author{
Ijay Narang\thanks{School of Computer Science, Georgia Institute of Technology,
\texttt{inarang3@gatech.edu}}}
\date{}

\begin{document}
\maketitle

\begin{abstract}

We give complexity-theoretic separations between $\#\mathrm{BIS}$ and $\#\mathrm{IS}$. Given a graph $G$, and a fugacity $\lambda>0$, the hard-core model is given by $Z_G(\lambda)=\sum_{I\in\mathcal I(G)}\lambda^{|I|}$, where $\mathcal I(G)$ denotes the independent sets of $G$. For $F\ge1$, an $F$-factor approximation to $Z_G(\lambda)$ is an estimate $\widehat Z$ satisfying $F^{-1}Z_G(\lambda)\le\widehat Z\le FZ_G(\lambda)$. Let $\mathcal G_\Delta$ denote the class of $\Delta$-regular graphs and $\mathcal G_{\Delta,\Delta}$ the class of $\Delta$-regular bipartite graphs. We give two separation results.
\begin{itemize}
    \item There exist an integer $\Delta\ge3$, a sufficiently large fugacity $\lambda$, and a constant $c>0$ so that the following hold:
    \begin{itemize}
        \item For every $\delta\in(0,1)$, deterministic $(1+\delta)$-factor approximation on $\mathcal G_{\Delta,\Delta}$ is possible in time $\exp(cn)\poly(n,1/\delta)$.
        \item Under the Gap-Exponential Time Hypothesis, there exists no deterministic algorithm running in time $\exp(2cn)\poly(n)$ on $\mathcal G_\Delta$ that gives an $\frac{11}{10}$-factor approximation.
    \end{itemize} 
    \item For every fixed $\Delta\ge3$ and $\varepsilon\in(0,1)$, for all sufficiently large fugacities $\lambda$, the following hold:
    \begin{itemize}
        \item There exists an algorithm which outputs a deterministic $(1+\varepsilon)^n$-factor approximation to $Z_G(\lambda)$ on $\mathcal G_{\Delta,\Delta}$ in polynomial time.
        \item Under the assumption $\mathrm P\ne\mathrm{NP}$, there exists no deterministic polynomial-time algorithm on $\mathcal G_\Delta$ that gives a $(1+\varepsilon)^n$-factor approximation.
    \end{itemize}
\end{itemize}
Both our algorithmic results follow from a new spanning-tree encoding of independent sets that enables an efficient representation of large independent sets.
\end{abstract}
\par\vfill
\paragraph{A note on OpenAI's disproof of ETH.}
On October 26, 2026, OpenAI announced a refutation of the Exponential Time Hypothesis. We have nevertheless decided to release this paper as originally planned, because we believe our spanning-tree encoding provides a useful technique for approximate counting and sampling at high fugacity, even beyond bipartite graphs (see Corollary \ref{cor:near-half-extension}). Our polynomial-time separation result remains unaffected by this development.
\clearpage
\section{Introduction}\label{sec:introduction}

The hard-core model originated in statistical physics in the study of lattice gas systems and has since received considerable attention in theoretical computer science. Given a graph $G$ and a fugacity parameter $\lambda > 0$, the hard-core model $\mu_{G, \lambda}$ defines a distribution over all independent sets of $G$. Let $\mathcal{I}(G)$ denote the set of all independent sets of $G$. The hard-core model has the Gibbs distribution
\[\mu_{G,\lambda}(I) = \frac{\lambda^{|I|}}{Z_G(\lambda)}, \qquad Z_G(\lambda) = \sum_{I \in \mathcal{I}(G)} \lambda^{|I|}\]
where the normalizing constant $Z_G(\lambda)$ is referred to as the partition function.

Exact evaluation of $Z_G(\lambda)$ is $\#P$-hard in general~\cite{Valiant79},
motivating the study of efficient approximation algorithms. For $F\ge1$,
an $F$-factor approximation to $Z_G(\lambda)$ is an estimate $\widehat Z$ satisfying
\begin{equation}\label{eq:factor-approximation}
 F^{-1}Z_G(\lambda)\leq\widehat Z\leq FZ_G(\lambda).
\end{equation}
An FPTAS is a deterministic algorithm that, given $\delta\in(0,1)$,
returns a $(1+\delta)$-factor approximation in time polynomial in
$|V(G)|$ and $1/\delta$. An FPRAS is a randomized algorithm with the same
running-time bound whose estimate satisfies this guarantee with probability
at least $3/4$. A $\delta$-approximate sampler for $\mu_{G,\lambda}$
produces an independent set whose distribution $\widehat\mu$ satisfies
$\TV(\widehat\mu,\mu_{G,\lambda})\leq\delta$. Such a sampler is efficient
if its running time is polynomial in $|V(G)|$ and $1/\delta$.
Here $\TV$ denotes total variation distance.

Let $\mathcal{G}^{\leq}_\Delta$ denote the set of all graphs of maximum
degree at most $\Delta$. The hard-core model on the infinite
$\Delta$-regular tree has a uniqueness threshold $\lambda_c(\Delta)$.
When $\lambda>\lambda_c(\Delta)$, there is no efficient approximate sampler
or FPRAS on $\mathcal{G}^{\leq}_\Delta$ unless
$\mathrm{NP}=\mathrm{RP}$~\cite{Sly10,SS14,gsvhc}; when
$\lambda<\lambda_c(\Delta)$, efficient approximate counting and sampling
are possible through correlation decay and the self-avoiding-walk tree
method~\cite{Weitz06}.

For bipartite graphs, however, the situation is less well understood. Let $\mathcal{G}^{\leq}_{\Delta, \Delta}$ denote the class of all bipartite graphs with maximum degree at most $\Delta$. Even at \(\lambda=1\), approximately counting independent sets in bipartite graphs gives rise to the canonical problem \(\BIS\), which is neither known to admit an FPRAS nor known to be NP-hard to approximate~\cite{DGGJ04,cai}. Moreover, many natural approximate counting problems have been shown to be \(\BIS\)-hard~\cite{DGGJ04}, making the complexity of \(\BIS\) a central open problem in approximate counting. More generally, Cai, Galanis, Goldberg, Guo, Jerrum, {\v{S}}tefankovi{\v{c}}, and Vigoda~\cite{cai} showed that, for every fixed \(\Delta\ge3\) and \(\lambda>\lambda_c(\Delta)\), approximating \(Z_G(\lambda)\) on graphs \(G\in\mathcal{G}^{\leq}_{\Delta,\Delta}\) is \(\BIS\)-hard. Thus, obtaining efficient approximation algorithms beyond uniqueness on general bipartite graphs remains a major open problem.

The study of approximate counting and sampling on bipartite graphs has led to several distinct algorithmic approaches. We shall briefly discuss the main approaches and their guarantees.

\paragraph{Correlation decay.}
The seminal correlation-decay approach of Weitz gives efficient approximate
counting on graphs of bounded maximum degree below the tree uniqueness
threshold~\cite{Weitz06}. For bipartite graphs, stronger guarantees are
possible: Liu and Lu established an FPTAS for $\BIS$ when the maximum degree
on one side is at most five, with no restriction on degrees on the other
side~\cite{LL15}. Chen, Liu, and Yin extended this direction to general
fugacities, characterizing the uniqueness condition under a degree bound on
one side and obtaining efficient sampling algorithms up to the corresponding
threshold~\cite{CLY23}.

\paragraph{The Polymer Model Framework.}
A second approach exploits the phase structure of the hard-core model at
high fugacity. On bipartite expanders, typical independent sets favor one
side of the bipartition, and deviations from these ordered phases can be
represented by polymers. Indeed, this idea was utilized by Jenssen, Keevash,
and Perkins to obtain efficient counting and sampling algorithms for
$\BIS$-hard problems on expander graphs~\cite{JKP20}.
Jenssen, Perkins, and Potukuchi subsequently combined graph containers with
cluster expansion methods to obtain counting algorithms under weaker
expansion conditions~\cite{JPP23}. Most recently, Narang and Perkins showed
that the cluster expansion framework also yields efficient counting and
sampling algorithms at high fugacity on bipartite spectral
expanders~\cite{NP}.

\paragraph{Un-bipartifying the model.}
On bipartite graphs, transitions between configurations favoring opposite
sides can create a bottleneck for Glauber dynamics~\cite{MWW09}. This
motivates decomposing the partition function to effectively
``un-bipartify'' the model. Kocurek, Oveis Gharan, and Tjowasi decomposed
the partition function into fixed slices to obtain efficient counting and
sampling at all fugacities on random regular bipartite graphs~\cite{KOGT26}. Narang and Perkins obtained efficient counting and sampling at all fugacities on bipartite spectral expanders using a discrete Hubbard--Stratonovich transform~\cite{NP} in combination with the Glauber mixing criterion of \cite{CCCYZ25}.
These methods primarily address the medium-fugacity regime; the all-fugacity guarantees combine them with high-fugacity polymer algorithms.
\newline

Despite this progress, the above results all require additional structural assumptions on the bipartite graph, meaning that no separation between the classes $\mathcal{G}_\Delta$ and $\mathcal{G}_{\Delta, \Delta}$ is known in full generality. 

\subsection{Our Contributions}

\subsubsection{Separation between $\#\mathrm{BIS}$ and $\#\mathrm{IS}$ under Gap-ETH}

Our results establish an exponential-time separation between regular bipartite graphs and general regular graphs for approximate counting under the Gap-Exponential Time Hypothesis (Gap-ETH), and for sampling under its randomized version. We use the following formulation of Gap-ETH ~\cite[Conjecture 2.5 and the subsequent discussion]{ChalermsookEtAl2017}.

\begin{definition}[Gap-ETH]
\label{def:gapeth}
There exist constants $D\geq1$, $\eta_0\in(0,1)$, and $\rho_0>0$ such that no deterministic algorithm running in time $O(\exp(\rho_0N))$ can distinguish between the following two cases for a $3$-CNF formula $\Phi$ on $N$ variables and $M\leq DN$ clauses:
\begin{itemize}
    \item $\Phi$ is satisfiable.
    \item Every assignment satisfies at most $(1-\eta_0)M$ clauses.
\end{itemize}

The randomized Gap-ETH asserts that no randomized algorithm with the same running time distinguishes
these cases with probability at least $2/3$ on every promised input.
\end{definition}

Gap-ETH was introduced in work of Dinur~\cite{Dinur2016} and Manurangsi and Raghavendra~\cite{ManurangsiRaghavendra2016} to study hardness of approximation in subexponential time. It has since yielded strong inapproximability results for problems such as Clique and Dominating Set~\cite{ChalermsookEtAl2017}.

We instead use it to analyze the hardness of approximate counting and sampling problems. Our first result is the following theorem.

\begin{theorem}
\label{thm:lower}
Assume Gap-ETH (Definition \ref{def:gapeth}). There exists an integer
$B\geq1$ and constants $\eta\in(0,1)$ and $\rho>0$ such that the following holds.
Fix any odd integer $\Delta>B+2$ and set
\[
C=3+3\Delta(\Delta+2),\qquad
c=\frac{\rho}{2C},\qquad
\lambda_{\mathrm H}=2^{4C/\eta}.
\]
For every fixed $\lambda\geq\lambda_{\mathrm H}$, no deterministic
algorithm produces an $11/10$-factor approximation to $Z_G(\lambda)$ in time $\exp(cn)$ on the class $\mathcal{G}_\Delta$.

Under randomized Gap-ETH, the constants can be chosen so that no sampler
with total variation error at most $1/10$ runs in time $\exp(cn)$ on the same graph class.
\end{theorem}

At sufficiently high fugacity, we complement this hardness result with algorithms on regular bipartite graphs whose exponential running-time rate can be made arbitrarily small.

\begin{theorem} \label{thm:upper}
Fix $\Delta\geq1$ and $\varepsilon\in(0,1)$. For every $\delta \in (0,1)$ and every fugacity $\lambda$ satisfying
\[
\lambda\geq\lambda_{\mathrm B}(\Delta,\varepsilon)
 :=\left[\frac{8e}{\varepsilon}\ln\frac4\varepsilon\right]^\Delta,
\]
there exists a deterministic algorithm that produces a $(1+\delta)$-factor approximation to $Z_G(\lambda)$ in time $(1+\varepsilon)^n\poly(n,1/\delta)$ on the class $\mathcal{G}_{\Delta, \Delta}.$ There also exists a $\delta$-approximate sampling algorithm for $\mu_{G, \lambda}$ with the same running time.
\end{theorem}

We note that though Theorem \ref{thm:upper} does not apply at $\lambda = 1$, it still applies to $\BIS$-hard problems (see Lemma \ref{lem:Delta-reg-bis-hard}). More interesting, however is the comparison between Theorem \ref{thm:upper} and the complexity of approximating the partition function of the hard-core model on general $\Delta$-regular graphs.

In particular, Theorems~\ref{thm:lower} and~\ref{thm:upper} separate approximation algorithms on regular bipartite graphs and general regular graphs at the same degree and fugacity. Indeed, fix $\Delta$, $c$, and $\lambda_{\mathrm H}$ as in
Theorem~\ref{thm:lower}, choose $\varepsilon\in(0,1)$ with
$\ln(1+\varepsilon)<c/2$, and fix $\lambda\geq
\max\{\lambda_{\mathrm H}, \lambda_{\mathrm B}(\Delta,\varepsilon)\}.$
An $11/10$-factor approximation on $\Delta$-regular bipartite graphs admits a deterministic $\exp(cn/2)\poly(n)$-time algorithm, whereas Gap-ETH rules out an $\exp(cn)\poly(n)$-time deterministic algorithm on all $\Delta$-regular graphs.

Having obtained an exponential-scale separation under Gap-ETH, we next ask whether a corresponding separation can be shown for polynomial-time algorithms. We show this in the next subsection.

\subsubsection{Separation Between $\#\mathrm{BIS}$ and $\#\mathrm{IS}$ for polynomial-time algorithms}

To analyze polynomial-time algorithms, we allow the approximation factor in
\eqref{eq:factor-approximation} to grow with the size of the graph, taking
$F=(1+\varepsilon)^n$ for $\varepsilon\in(0,1)$.
Our next result shows that at sufficiently large fugacity, one can attain a sufficiently good exponential approximation on the class $\mathcal{G}_{\Delta, \Delta}$ in polynomial time.

\begin{theorem}\label{thm:poly-alg}
Fix $\Delta\ge 3$ and $\varepsilon\in(0,1)$. For every fixed $\lambda\ge\varepsilon^{-\Delta}$,
there is a deterministic polynomial-time $(1+\varepsilon)^n$-factor
approximation to $Z_G(\lambda)$ on the class $\mathcal{G}_{\Delta, \Delta}$.
\end{theorem}

We now complement the above result with the corresponding result on the class of graphs $\mathcal{G}_\Delta$.

\begin{theorem}\label{thm:poly-hard}
Fix $\Delta \ge 3$, $\varepsilon \in (0,1)$, and set $C_\Delta=200(\Delta^2-7)$.
For every fixed $\lambda>[2(1+\varepsilon)^2]^{C_\Delta}$, producing a
$(1+\varepsilon)^n$-factor approximation on the class $\mathcal{G}_\Delta$ is $\mathrm{NP}$-hard.
\end{theorem}

\subsection{Discussion}\label{sec:intro-discussion}

A driving question of this paper was separation between optimization and approximate sampling. We observe that in the case of the hard-core model, one can construct simple examples which establish separation between these problems.

Indeed, let $\mathcal{C}$ be some class of graphs on which approximating $Z_G(\lambda)$ is hard. For each $G\in\mathcal{C}$, form $G'$ by attaching a new leaf to every vertex. Let $\mathcal{C}'$ be the resulting graph class.

It is easy to see that for all $G' \in \mathcal{C'}$, $\alpha(G') = |V(G)|$. However, for the partition function, we have that
\[Z_{G'}(\lambda) = \sum_{I \in \mathcal{I}(G)} \lambda^{|I|} \cdot (1 + \lambda)^{|V(G)| - |I|} =(1+\lambda)^{|V(G)|}
\sum_{I\in\mathcal I(G)}
\left(\frac{\lambda}{1+\lambda}\right)^{|I|}
=(1+\lambda)^{|V(G)|}\,
Z_G\!\left(\frac{\lambda}{1+\lambda}\right)\]

\noindent and thus hardness of approximation persists. However, the  algorithm and analysis underlying Theorems~\ref{thm:upper} and \ref{thm:poly-alg} are simple and follow from enumerating independent sets with few defects along a spanning tree. Here, a defect refers to an edge with both endpoints unoccupied. As such, the simplicity of the approach enables its extension to regular graphs where $\frac{\alpha(G)}{n} \approx \frac{1}{2}$.

For $\eta\ge0$, let
\[
    \mathcal G_\Delta(\eta)
    :=\left\{
        G\in\mathcal G_\Delta:
        \alpha(G)\ge\left(\frac12-\eta\right)|V(G)|
      \right\}.
\]

\begin{corollary}\label{cor:near-half-extension}
Fix $\Delta\ge1$ and $\varepsilon\in(0,1)$.
There exist $\eta>0$ and $\lambda_0>1$ such that, for every
$\delta\in(0,1)$ and every fixed fugacity
$\lambda\ge\lambda_0$, there exists a deterministic algorithm
that produces a $(1+\delta)$-factor approximation to $Z_G(\lambda)$
in time $(1+\varepsilon)^n\poly(n,1/\delta)$ on the class
$\mathcal G_\Delta(\eta)$.
There also exists a $\delta$-approximate sampling algorithm
for $\mu_{G,\lambda}$ with the same running time.
\end{corollary}

A natural question is whether the condition that $\frac{\alpha(G)}{n} \approx \frac{1}{2}$ can be replaced by a suitable assumption on optimization. The preceding example shows that efficient optimization
alone is insufficient, but is artificial in the sense that we made
optimization easy by adding vertices. This motivates asking whether a
stronger optimization guarantee can yield fast approximate counting
and sampling.

\begin{question}
Given a polynomial-time oracle for maximum independent set
on every induced subgraph, can we obtain a $(1+\delta)$-factor approximation
to $Z_G(\lambda)$ and a $\delta$-approximate sampler for $\mu_{G,\lambda}$
on the class $\mathcal G_\Delta$ in time
$(1+\varepsilon)^n\poly(n,1/\delta)$, for every
$\delta\in(0,1)$ and $\varepsilon>0$ at sufficiently large fugacity?
\end{question}

The above formulation is natural due to the role of self-reducibility in approximate sampling and counting. Jerrum, Valiant, and Vazirani,
and subsequently Sinclair and Jerrum, established that randomized
approximate counting and almost-uniform sampling are equivalent for self-reducible problems~\cite{JVV86,SJ89}.
These reductions rely on access to counting or sampling on conditioned
instances. On the hard-core model, this corresponds to induced subgraphs.

One natural setting for the above question is perfect graphs. Every induced subgraph of a perfect graph is perfect, and furthermore, maximum independent set can be solved in polynomial time on this class~\cite{GLS1984}.

\begin{question}
Does Corollary~\ref{cor:near-half-extension} extend to all
$\Delta$-regular perfect graphs, without the assumption that
$\frac{\alpha(G)}{n}$ is close to $\frac{1}{2}$?
\end{question}

\section{Bipartite versus general graphs under Gap-ETH}
\label{sec:gapeth-separation}

\subsection{Approximation algorithm for bipartite graphs}
\label{sec:upper}

Before we describe the algorithm underlying Theorem \ref{thm:upper}, we introduce the ``cut'' representation of the hard-core model on bipartite graphs. Let $G=(L\sqcup R,E)$ be a $\Delta$-regular bipartite graph. For an independent set $I$, define
\[
z_v=\begin{cases}
\mathbf1_{\{v\in I\}},&v\in L,\\
1-\mathbf1_{\{v\in I\}},&v\in R.
\end{cases}
\]

These labels express the standard bipartite order for the hard-core model
as coordinatewise order~\cite[Section~3]{vandenBergSteif1994}.
In this encoding, the measure satisfies the FKG lattice condition and hence is positively associated by the FKG inequality~\cite{FKG1971}.

Let $\mathcal F_G$ denote the set of binary labelings which correspond to a feasible independent set, and write
$c(z)=\sum_{uv\in E}\mathbf{1}_{\{z_u\neq z_v\}}$.
The only permitted disagreeing pair is $(0,1)$, corresponding to
neither endpoint belonging to $I$. Since $I$ is independent and
$G$ is $\Delta$-regular,
\[
    c(z)=|E|-\Delta|I|=\Delta(n/2-|I|).
\]
Expressing $\lambda^{|I|}$ in terms of $c(z)$ therefore gives
\[
    Z_G(\lambda)
    =\lambda^{n/2}
      \sum_{z\in\mathcal F_G}
      e^{-(\log\lambda/\Delta)c(z)}.
\]

Call an edge a \emph{defect} of an independent set $I$ if neither endpoint belongs to $I$.
We first assume that $G$ is connected and handle disconnected graphs at the end of the proof. Fix a rooted spanning tree $T$ of $G$. Our key observation is that the root label and the set of defective edges of $T$ determine all vertex labels (for a feasible independent set). Thus, we can simply enumerate both root labels and all sets of at most $K=\lfloor a(n-1)\rfloor$ defective tree
edges and discard labelings that do not encode independent sets. 

To see why this works, set $a=\varepsilon/(4\ln(4/\varepsilon))$ and $F=1+\delta$. Letting $d_T(I)$ denote the number of edges of $T$ with neither endpoint in $I$, the algorithm then follows immediately by the following simple chain of inequalities.
\[
\begin{aligned}
\widehat{Z}
&=\sum_{\substack{I\in\mathcal I(G)\\ d_T(I)\le K}}\lambda^{|I|} \ge Z_G(\lambda)
-2\lambda^{n/2}\sum_{j=K+1}^{n-1}
\binom{n-1}{j}\lambda^{-j/\Delta} \ge (1-2^{-K})Z_G(\lambda) \ge F^{-1}Z_G(\lambda).
\end{aligned}
\]
The last inequality holds whenever $2^{-K}\le1-F^{-1}$.
Since we only discard positive terms, $\widehat{Z}\le Z_G(\lambda)$,
so $\widehat{Z}$ is an $F$-factor approximation.
The formal description of the algorithm is given below.

\begin{algorithm}[H]
\caption{Approximation by spanning-tree enumeration}
\label{alg:tree-enumeration}
\begin{algorithmic}[1]
\Require A connected $n$-vertex $\Delta$-regular bipartite graph
$G=(L\sqcup R,E)$, approximation factor $F>1$, a fugacity $\lambda>0$, and $a\in(0,1/2)$ with
$\lambda^{-1/\Delta}\leq a/(2e)$ and $2^{-\lfloor a(n-1)\rfloor}\leq1-F^{-1}$.
\Ensure $F^{-1}Z_G(\lambda)\leq\widehat Z\leq Z_G(\lambda)$.
\State Compute a rooted spanning tree $T$ of $G$.
\State Set $K\gets\lfloor a(n-1)\rfloor$ and $\widehat Z\gets0$.
\For{$b\in\{0,1\}$}
  \ForAll{$S\subseteq E(T)$ with $|S|\leq K$}
    \State Assign the root label $b$ and flip labels across precisely the edges in $S$.
    \State $I\gets\{\ell\in L:z_\ell=1\}\cup\{v\in R:z_v=0\}$.
    \If{$I$ is independent in $G$}
      \State $\widehat Z\gets\widehat Z+\lambda^{|I|}$.
    \EndIf
  \EndFor
\EndFor
\State \Return $\widehat Z$.
\end{algorithmic}
\end{algorithm}

\begin{proof}[Proof of Theorem~\ref{thm:upper}]

Put $N=n-1$, $t=\lambda^{-1/\Delta}$, and $K=\lfloor aN\rfloor$,
where $a\in(0,1/2)$ satisfies $t\leq a/(2e)$. Set $F=1+\delta$.
We may assume $n$ is sufficiently large that $2^{-K}\le\delta/(1+\delta)$,
since smaller instances are absorbed by the $\poly(1/\delta)$ factor.
Algorithm~\ref{alg:tree-enumeration} tests at most
$2\sum_{j=0}^K\binom Nj$ labelings, each in polynomial time. Let $h(a)=-a\ln a-(1-a)\ln(1-a)$ denote the binary entropy function. We have
\[
\sum_{j=0}^{\lfloor aN\rfloor}\binom Nj
\leq a^{-aN}(1-a)^{-(1-a)N}
\sum_{j=0}^N\binom Nj a^j(1-a)^{N-j}
=e^{Nh(a)}
\]
giving running time
$e^{nh(a)}\poly(n,1/\delta)$.
To obtain the stated bound, set $L_0=\ln(4/\varepsilon)$ and
$a_0=\varepsilon/(4L_0)$. The fugacity assumption gives
$t\leq a_0/(2e)$, while
\[
h(a_0)\leq a_0\ln(e/a_0)
=\frac{\varepsilon}{4L_0}(1+L_0+\ln L_0)
\leq\frac{\varepsilon}{2}<\ln(1+\varepsilon).
\]
Taking $a=a_0$ therefore gives running time
$(1+\varepsilon)^n\poly(n,1/\delta)$.

Now we show correctness. The root label and the set of tree defects
uniquely determine a labeling, so every retained independent set is
counted exactly once. A labeling with $j$ tree defects has $c(z)\geq j$,
and there are $2\binom Nj$ such labelings. Since the two constant
labelings are feasible, $Z_G(\lambda)\geq2\lambda^{n/2}$. We can thus bound the omitted weight by
\[
0\leq\frac{Z_G(\lambda)-\widehat Z}{Z_G(\lambda)}
\leq\sum_{j=K+1}^N\binom Nj t^j
\leq\sum_{j=K+1}^N\left(\frac{eNt}{j}\right)^j
\leq\sum_{j=K+1}^{\infty}2^{-j}
=2^{-K}.
\]
Since $2^{-K}\le\delta/(1+\delta)$, this gives
$Z_G(\lambda)/(1+\delta)\leq\widehat Z\leq Z_G(\lambda)$,
and hence the claimed $(1+\delta)$-factor approximation.

For sampling, run the enumeration with factor $F=1+\delta/2$, and let $\mathcal A$ be the kept independent sets. Compute
$W=\sum_{I\in\mathcal A}\lambda^{|I|}$ and sample $I\in\mathcal A$
with probability $\lambda^{|I|}/W$. The resulting law is
$\nu=\mu_{G,\lambda}(\,\cdot\mid\mathcal A)$, so
\[
\TV(\nu,\mu_{G,\lambda})
=1-\mu_{G,\lambda}(\mathcal A)\leq\frac{\delta}{2+\delta}\leq\delta/2.
\]
This weighted selection can be implemented in a second enumeration
pass with at most $\delta/2$ additional total variation error,
preserving the running time.

For a graph with $r$ connected components, run the connected counting
algorithm on each component with factor $1+\delta/(2r)$ and multiply
the estimates. The resulting factor is at most
\[
 \left(1+\frac{\delta}{2r}\right)^r\le e^{\delta/2}\le1+\delta.
\]
The sampling argument then follows in the exact same manner.
\end{proof}

\subsection{Hardness on general graphs}
\label{sec:lower}

We reduce gap 3-SAT to the problem of approximating $Z_G(\lambda)$ at large fugacity. The proof follows from standard reduction techniques, but adapted to the counting setting.

\begin{proof}[Proof of Theorem~\ref{thm:lower}]
We shall reduce from the bounded-occurrence form of Gap-ETH, which is equivalent to Gap-ETH up to
a change of constants~\cite[Section 5.4]{ChalermsookEtAl2017}.

This formulation states that there exist an integer $B\geq1$ and constants $\eta\in(0,1)$ and $\rho>0$ such that no deterministic algorithm running in time $O(\exp(\rho M))$ distinguishes satisfiable $3$-CNF formulas with $M$ clauses from those in which every assignment satisfies at most $(1-\eta)M$ clauses, even when each variable has at most $B$ literal occurrences. Under randomized Gap-ETH, the corresponding statement holds
for randomized algorithms with success probability at least $2/3$.

Fix these constants, and let $\Phi$ be such a bounded-occurrence instance with variables $x_1,\ldots,x_N$ and clauses $C_1,\ldots,C_M$, where $C_i=(\ell_{i,1}\vee\ell_{i,2}\vee\ell_{i,3})$ and each literal $\ell_{i,j}$ is either $x_t$ or $\neg x_t$ for some $t\in[N]$.

Construct a graph $H$ with vertex set
$V(H)=\{v_{i,j}:i\in[M],\,j\in[3]\}$, where $v_{i,j}$ represents
the literal occurrence $\ell_{i,j}$. We connect two distinct vertices $v_{i,j}$ and $v_{i',j'}$ if $i=i'$, or if one of the literals
$\ell_{i,j}$ and $\ell_{i',j'}$ is $x_t$ and the other is $\neg x_t$
for some $t\in[N]$. Thus each clause induces a triangle, and
$|V(H)|=3M$ and $\deg_H(v)\leq B+2$.

An independent set selects at most one literal occurrence per clause
and never selects both a variable and its negation, so its selections
extend to an assignment satisfying at least that many clauses.
Conversely, choosing one true literal occurrence from each satisfied
clause gives an independent set. Hence $\alpha(H)=\MaxSAT(\Phi)$.

To make $H$ regular, we use the following graph $Q$. Start with
the complete graph on vertices $\{p,a,b,u_1,\ldots,u_{\Delta-1}\}$,
and delete the edges $pa$, $pb$, and
$u_{2i-1}u_{2i}$ for $i=1,\ldots,(\Delta-1)/2$.
Then $p$ has degree $\Delta-1$, while every other vertex has
degree $\Delta$. The deleted edges form a three-vertex path and
a matching, so no three vertices of $Q$ are pairwise nonadjacent.
Since $\{u_1,u_2\}$ is independent and avoids $p$, we have
$\alpha(Q)=\alpha(Q-p)=2$.

Attach
$\Delta-\deg_H(v)$ fresh copies of $Q$ to each $v\in V(H)$, joining
$v$ to $p$ in each copy. This produces a $\Delta$-regular graph $G$
in polynomial time. Writing $s\leq3\Delta M$ for the number of copies,
we have
\[
n=3M+(\Delta+2)s\leq CM,
\qquad
\alpha(G)=\MaxSAT(\Phi)+2s,
\]
as every independent set of $H$ extends by two nonport vertices from each copy, and no copy contributes more than two vertices.

For counting, put $k=M+2s$ and use
$\lambda^{\alpha(G)}\leq Z_G(\lambda)\leq2^n\lambda^{\alpha(G)}$.
If $\Phi$ is satisfiable, then $Z_G(\lambda)\geq\lambda^k$. If every assignment to $\Phi$ satisfies at most $(1-\eta)M$ clauses, then
\[
Z_G(\lambda)\leq2^n\lambda^{k-\eta M}
\leq2^{CM}\lambda^{-\eta M}\lambda^k
\leq2^{-3CM}\lambda^k,
\]
since $\lambda\geq2^{4C/\eta}$. Thus an $11/10$-factor approximation
is at least $(10/11)\lambda^k>3\lambda^k/4$ in the first case and at most
$(11/10)2^{-3CM}\lambda^k<3\lambda^k/4$ in the second.

An $\exp(cn)\poly(n)$-time approximation algorithm would therefore
solve the promise problem in time
\[
\exp(cn)\poly(n)\leq\exp(cCM)\poly(M)
=\exp(\rho M/2)\poly(M)=O(\exp(\rho M)),
\]
contradicting Gap-ETH. The same argument applies to randomized
counting under randomized Gap-ETH.

For sampling, fix $\gamma\in(\eta/2,\eta)$ and accept an
output $I$ if it is independent and $|I|>k-\gamma M$. In the
unsatisfiable promised case, $\alpha(G)\leq k-\eta M$, so the test
always rejects. In the satisfiable case, $Z_G(\lambda)\geq\lambda^k$
and
\[
\Pr_{I\sim\mu_{G,\lambda}}[|I|\leq k-\gamma M]
\leq2^n\lambda^{-\gamma M}\leq2^{-CM}.
\]
A sampler with total variation error at most $1/10$ therefore gives
acceptance probability at least $1-2^{-CM}-1/10>2/3$. The test takes
polynomial time, so the preceding running-time bound contradicts
randomized Gap-ETH.
\end{proof}

\section{Bipartite versus general graphs under \texorpdfstring{$\mathrm P\ne\mathrm{NP}$}{P != NP}}
\label{sec:coarse-separation-proof}

\subsection{General approximation algorithm}
\label{sec:coarse-approximation}

We begin with a bound for connected regular graphs whose independence
number is known, using the same spanning-tree encoding as in the
previous section.

\begin{lemma}\label{thm:coarse-general}
Fix $\Delta\ge1$ and a fugacity $\lambda\ge1$. Let $G$ be
a connected $n$-vertex $\Delta$-regular graph, and suppose
$\alpha=\alpha(G)$ is known. Put
\[
 L=\lambda^\alpha,\qquad
 U=2\lambda^{n/2}(1+\lambda^{-1/\Delta})^{n-1}.
\]
Then $L\le Z_G(\lambda)\le U$, and the geometric midpoint
$\widehat Z=\sqrt{LU}$ is an $F$-factor approximation with
\[
 F=\sqrt{\frac UL}
 =\sqrt2\,\lambda^{(n/2-\alpha)/2}
 (1+\lambda^{-1/\Delta})^{(n-1)/2}.
\]
\end{lemma}
\begin{proof}
A maximum independent set contributes $L$ to the partition function,
giving the lower bound. For the upper bound, choose a rooted spanning
tree $T$ of $G$.

The tree has $n-1$ edges, and the root occupancy and tree defects
determine an independent set. Thus at most $2\binom{n-1}{j}$
independent sets have $j$ tree defects. Let $d(I)$ be the total number
of defects of $I$. Regularity gives $d(I)=\Delta(n/2-|I|)$, and a set
with $j$ tree defects has $d(I)\ge j$. Since $\lambda\ge1$, its weight
is at most
$\lambda^{n/2}(\lambda^{-1/\Delta})^j$. Summing over $j$ gives
\begin{equation}\label{eq:coarse-tree-bound}
 Z_G(\lambda)\le
  2\lambda^{n/2}\sum_{j=0}^{n-1}\binom{n-1}{j}
 (\lambda^{-1/\Delta})^j
 =2\lambda^{n/2}(1+\lambda^{-1/\Delta})^{n-1}.
\end{equation}
For every $Z\in[L,U]$, both $Z/\sqrt{LU}$ and $\sqrt{LU}/Z$ are at most $\sqrt{U/L}$, giving the approximation guarantee.
\end{proof}

The proof of Theorem \ref{thm:poly-alg} now follows immediately.

\begin{proof}[Proof of Theorem~\ref{thm:poly-alg}]
As the hard-core model factorizes over connected components, we may assume that $G$
is connected. Put $L=2\lambda^{n/2}$. The two bipartition classes have
size $n/2$, giving $Z_G(\lambda)\ge L$.
Together with \eqref{eq:coarse-tree-bound} and
$\lambda^{-1/\Delta}\le\varepsilon$, this gives
\[
 L\le Z_G(\lambda)
 \le L(1+\lambda^{-1/\Delta})^{n-1}
 \le L(1+\varepsilon)^n.
\]
Taking $F=(1+\varepsilon)^n$, we have $Z_G(\lambda)/F\le L\le Z_G(\lambda)$,
so outputting $L$ gives an $F$-factor approximation. For a disconnected
graph, the component factors multiply to $(1+\varepsilon)^{\sum_i n_i}
=(1+\varepsilon)^n$, where $n_i$ are the component sizes.
\end{proof}

\subsection{Hardness on general graphs}
\label{sec:coarse-hardness}

Our proof will follow by transferring a gap in the independence number of cubic graphs to $\Delta$-regular graphs. At sufficiently large fugacity, the resulting gap in the partition function will exceed the allowed approximation factor.

\begin{proof}[Proof of Theorem~\ref{thm:poly-hard}]
We use the following consequence of Chleb\'ik--Chleb\'ikov\'a
\cite[Theorems 2, 16--17 and Corollary 18]{ChlebikChlebikova2006}: there exist fixed $a>b>0$ with $a-b\ge5/1936$ such that distinguishing
$\alpha(G)>an$ from $\alpha(G)<bn$ is NP-hard on simple cubic graphs.
For completeness, their fixed-gadget construction has
$n=t(2r+16)$ and $\alpha(G)=t(r+6+\mathrm{OPT})$, where
$r\le40.4$ is fixed. Applying their source gap
$\mathrm{OPT}>7/8$ versus $\mathrm{OPT}<5/8$ gives
\[
 a=\frac{r+55/8}{2r+16},\qquad
 b=\frac{r+53/8}{2r+16},\qquad
 a-b=\frac1{4(2r+16)}\ge\frac5{1936}>\frac1{400}.
\]
We raise the degree from $3$ to $\Delta$ using a gadget with
two attachment vertices. Let $h=\Delta+3\ge6$, and obtain $Q$ by
adding the chord $03$ to the cycle $0,1,\ldots,h-1,0$ and then
taking the complement.
Vertices $0,3$ are ports of degree $\Delta-1$ in $Q$ and all others have degree $\Delta$. The cycle with the added chord is
triangle-free, since $0,3$ have disjoint cycle-neighbor sets
$\{1,h-1\},\{2,4\}$. Thus $\alpha(Q)\le2$, and $\{1,2\}$ is an
independent set of size two avoiding both ports.

As the cubic graph $G$ has even cardinality, pair its vertices arbitrarily. For each pair $x,y$, attach $\Delta-3$ fresh copies of $Q$, joining port $0$ to $x$ and port
$3$ to $y$ in each copy. The resulting graph $H$ is simple and
$\Delta$-regular. With $t=n(\Delta-3)/2$ gadgets,
\[
 N:=|V(H)|=n+(\Delta+3)t=\frac{\Delta^2-7}{2}n,
 \qquad \alpha(H)=\alpha(G)+2t.
\]
Indeed, any independent set uses at most two vertices per gadget,
and any independent set of $G$ extends by selecting $\{1,2\}$ in
every gadget. Now, set $A'=an+2t$, $B'=bn+2t$. The promises become
$\alpha(H)>A'$ and $\alpha(H)<B'$, with
$A'-B'=(a-b)n>n/400=N/C_\Delta$. For fixed $\Delta$, this construction has
linear size and takes polynomial time.

It remains to recover this gap from an approximate count.
Since $\lambda>1$, a maximum independent set contributes
$\lambda^{\alpha(H)}$, and each of at most $2^N$ independent sets
contributes no more. Therefore
$\lambda^{\alpha(H)}\le Z_H(\lambda)\le2^N\lambda^{\alpha(H)}$.
An $F=(1+\varepsilon)^N$-factor approximation consequently satisfies
\[
 \widehat Z>\lambda^{A'}/F\quad\text{in the yes case},\qquad
 \widehat Z<F2^N\lambda^{B'}\quad\text{in the no case}.
\]
The fugacity assumption makes these intervals disjoint, since
\[
 (A'-B')\log\lambda>
 \frac{N}{C_\Delta}\log\lambda>
 N\bigl(\log2+2\log(1+\varepsilon)\bigr).
\]
To implement the comparison for any fixed $\lambda$, note that
$A'$, $B'$, and $N$ are fixed multiples of $n$, so
$a'=A'/N$ and $b'=B'/N$ are independent of the input. Choose a fixed $q$ satisfying
\[
 2(1+\varepsilon)\lambda^{b'}<q<
 \frac{\lambda^{a'}}{1+\varepsilon},
\]
which exists by the preceding inequality. Raising to the $N$th power gives
\[
 F2^N\lambda^{B'}<q^N<\lambda^{A'}/F.
\]
Thus every valid estimate lies above $q^N$ in the yes case and below
it in the no case. Since $q$ is fixed, the threshold
$q^N$ can be computed in polynomial time. Accepting exactly when
$\widehat Z>q^N$ therefore decides the original NP-hard promise problem.
\end{proof}

\section{Acknowledgements}
The author thanks Pedro Paredes, as the conjecture that $\#\mathrm{BIS}$ could be separated from $\#\mathrm{IS}$ under a variant of the Exponential Time Hypothesis first arose during a conversation while visiting Princeton. The author is also grateful to Will Perkins for helpful feedback on initial versions of this work and for many insightful conversations about $\#\mathrm{BIS}$ over the last year.

\paragraph{Statement of AI Use.}
The positive algorithmic results and proofs, based on leveraging spanning trees, are entirely due to the author. ChatGPT 6 Astra and Codex were used to assist with the hardness proofs and with typesetting, polishing, and verifying the correctness of the results. The author assumes responsibility for all content.

\bibliographystyle{alpha}
\bibliography{ref}

@article{cai,
  author        = {Cai, Jin-Yi and Galanis, Andreas and Goldberg, Leslie Ann and Guo, Heng and
                   Jerrum, Mark and {\v{S}}tefankovi{\v{c}}, Daniel and Vigoda, Eric},
  title         = {{\#BIS}-hardness for 2-spin systems on bipartite bounded degree graphs in
                   the tree non-uniqueness region},
  journal       = {Journal of Computer and System Sciences},
  volume        = {82},
  number        = {5},
  pages         = {690--711},
  year          = {2016},
  doi           = {10.1016/j.jcss.2015.11.009}
}

@misc{ChalermsookEtAl2017,
  author        = {Chalermsook, Parinya and Cygan, Marek and Kortsarz, Guy and Laekhanukit,
                   Bundit and Manurangsi, Pasin and Nanongkai, Danupon and Trevisan, Luca},
  title         = {From {Gap-ETH} to {FPT}-inapproximability: {Clique}, {Dominating Set}, and
                   more},
  year          = {2017},
  eprint        = {1708.04218},
  archivePrefix = {arXiv},
  url           = {https://arxiv.org/abs/1708.04218},
  howpublished  = {\href{https://arxiv.org/abs/1708.04218}{arXiv:1708.04218}},
  note          = {See Section 5.4 for the bounded-occurrence formulation of Gap-ETH}
}

@inproceedings{CLY23,
  author        = {Chen, Xiaoyu and Liu, Jingcheng and Yin, Yitong},
  title         = {Uniqueness and rapid mixing in the bipartite hardcore model (extended
                   abstract)},
  booktitle     = {Proceedings of the 64th Annual IEEE Symposium on Foundations of Computer
                   Science (FOCS 2023)},
  pages         = {1991--2005},
  publisher     = {IEEE},
  year          = {2023},
  doi           = {10.1109/FOCS57990.2023.00121}
}

@inproceedings{CCCYZ25,
  author        = {Chen, Xiaoyu and Chen, Zejia and Chen, Zongchen and Yin, Yitong and Zhang,
                   Xinyuan},
  title         = {Rapid mixing on random regular graphs beyond uniqueness},
  booktitle     = {Proceedings of the 66th Annual IEEE Symposium on Foundations of Computer
                   Science (FOCS 2025)},
  pages         = {2170--2193},
  publisher     = {IEEE},
  year          = {2025},
  doi           = {10.1109/FOCS63196.2025.00115}
}

@article{ChlebikChlebikova2006,
  author        = {Chleb{\'i}k, Miroslav and Chleb{\'i}kov{\'a}, Janka},
  title         = {Complexity of approximating bounded variants of optimization problems},
  journal       = {Theoretical Computer Science},
  volume        = {354},
  number        = {3},
  pages         = {320--338},
  year          = {2006},
  doi           = {10.1016/j.tcs.2005.11.029},
  url           = {https://pure.port.ac.uk/ws/portalfiles/portal/1887750/3DM_JOURNAL_revision_old.pdf},
  note          = {Theorem numbering cited in the text follows the linked
                   \href{https://pure.port.ac.uk/ws/portalfiles/portal/1887750/3DM_JOURNAL_revision_old.pdf}{author manuscript}}
}

@misc{Dinur2016,
  author        = {Dinur, Irit},
  title         = {Mildly exponential reduction from gap {3SAT} to polynomial-gap label-cover},
  year          = {2016},
  url           = {https://eccc.weizmann.ac.il/report/2016/128/},
  howpublished  = {\href{https://eccc.weizmann.ac.il/report/2016/128/}{Electronic Colloquium
                   on Computational Complexity, Report TR16-128}}
}

@article{DGGJ04,
  author        = {Dyer, Martin and Goldberg, Leslie Ann and Greenhill, Catherine and Jerrum,
                   Mark},
  title         = {The relative complexity of approximate counting problems},
  journal       = {Algorithmica},
  volume        = {38},
  number        = {3},
  pages         = {471--500},
  year          = {2004},
  doi           = {10.1007/s00453-003-1073-y}
}

@article{FKG1971,
  author        = {Fortuin, C. M. and Kasteleyn, P. W. and Ginibre, J.},
  title         = {Correlation inequalities on some partially ordered sets},
  journal       = {Communications in Mathematical Physics},
  volume        = {22},
  pages         = {89--103},
  year          = {1971},
  doi           = {10.1007/BF01651330}
}

@article{gsvhc,
  author        = {Galanis, Andreas and {\v{S}}tefankovi{\v{c}}, Daniel and Vigoda, Eric},
  title         = {Inapproximability of the partition function for the antiferromagnetic
                   {Ising} and hard-core models},
  journal       = {Combinatorics, Probability and Computing},
  volume        = {25},
  number        = {4},
  pages         = {500--559},
  year          = {2016},
  doi           = {10.1017/S0963548315000401}
}

@article{GLS1984,
  author        = {Gr{\"o}tschel, Martin and Lov{\'a}sz, L{\'a}szl{\'o} and Schrijver,
                   Alexander},
  title         = {Polynomial algorithms for perfect graphs},
  journal       = {Annals of Discrete Mathematics},
  volume        = {21},
  pages         = {325--356},
  year          = {1984},
  url           = {https://www.zib.de/userpage/groetschel/pubnew/paper/groetschellovaszschrijver1984b.pdf}
}

@article{JKP20,
  author        = {Jenssen, Matthew and Keevash, Peter and Perkins, Will},
  title         = {Algorithms for {{\#BIS}}-hard problems on expander graphs},
  journal       = {SIAM Journal on Computing},
  volume        = {49},
  number        = {4},
  pages         = {681--710},
  year          = {2020},
  doi           = {10.1137/19M1286669}
}

@article{JPP23,
  author        = {Jenssen, Matthew and Perkins, Will and Potukuchi, Aditya},
  title         = {Approximately counting independent sets in bipartite graphs via graph
                   containers},
  journal       = {Random Structures \& Algorithms},
  volume        = {63},
  number        = {1},
  pages         = {215--241},
  year          = {2023},
  doi           = {10.1002/rsa.21145}
}

@article{JVV86,
  author        = {Jerrum, Mark and Valiant, Leslie G. and Vazirani, Vijay V.},
  title         = {Random generation of combinatorial structures from a uniform distribution},
  journal       = {Theoretical Computer Science},
  volume        = {43},
  pages         = {169--188},
  year          = {1986},
  doi           = {10.1016/0304-3975(86)90174-X},
  url           = {https://doi.org/10.1016/0304-3975(86)90174-X}
}

@misc{KOGT26,
  author        = {Kocurek, Nicholas and {Oveis Gharan}, Shayan and Tjowasi, Dante},
  title         = {Sampling from the hardcore model on random regular bipartite graphs above
                   the uniqueness threshold},
  year          = {2026},
  eprint        = {2604.21847},
  archivePrefix = {arXiv},
  url           = {https://arxiv.org/abs/2604.21847},
  howpublished  = {\href{https://arxiv.org/abs/2604.21847}{arXiv:2604.21847}}
}

@inproceedings{LL15,
  author        = {Liu, Jingcheng and Lu, Pinyan},
  title         = {{FPTAS} for {{\#BIS}} with degree bounds on one side},
  booktitle     = {Proceedings of the 47th Annual ACM Symposium on Theory of Computing (STOC
                   2015)},
  pages         = {549--556},
  publisher     = {ACM},
  year          = {2015},
  doi           = {10.1145/2746539.2746598}
}

@inproceedings{ManurangsiRaghavendra2016,
  author        = {Manurangsi, Pasin and Raghavendra, Prasad},
  title         = {A birthday repetition theorem and complexity of approximating dense {CSPs}},
  booktitle     = {44th International Colloquium on Automata, Languages, and Programming
                   (ICALP 2017)},
  series        = {Leibniz International Proceedings in Informatics (LIPIcs)},
  volume        = {80},
  pages         = {78:1--78:15},
  publisher     = {Schloss Dagstuhl--Leibniz-Zentrum f{\"u}r Informatik},
  year          = {2017},
  doi           = {10.4230/LIPIcs.ICALP.2017.78},
  url           = {https://drops.dagstuhl.de/entities/document/10.4230/LIPIcs.ICALP.2017.78}
}

@article{MWW09,
  author        = {Mossel, Elchanan and Weitz, Dror and Wormald, Nicholas},
  title         = {On the hardness of sampling independent sets beyond the tree threshold},
  journal       = {Probability Theory and Related Fields},
  volume        = {143},
  number        = {3--4},
  pages         = {401--439},
  year          = {2009},
  doi           = {10.1007/s00440-007-0131-9}
}

@misc{NP,
  author        = {Narang, Ijay and Perkins, Will},
  title         = {The hard-core model on bipartite spectral expanders: Counting and sampling
                   at all fugacities},
  year          = {2026},
  eprint        = {2608.03848},
  archivePrefix = {arXiv},
  url           = {https://arxiv.org/abs/2608.03848},
  howpublished  = {\href{https://arxiv.org/abs/2608.03848}{arXiv:2608.03848}}
}

@article{SJ89,
  author        = {Sinclair, Alistair and Jerrum, Mark},
  title         = {Approximate counting, uniform generation and rapidly mixing {Markov} chains},
  journal       = {Information and Computation},
  volume        = {82},
  number        = {1},
  pages         = {93--133},
  year          = {1989},
  doi           = {10.1016/0890-5401(89)90067-9},
  url           = {https://people.eecs.berkeley.edu/~sinclair/approx.pdf}
}

@inproceedings{Sly10,
  author        = {Sly, Allan},
  title         = {Computational transition at the uniqueness threshold},
  booktitle     = {Proceedings of the 51st Annual IEEE Symposium on Foundations of Computer
                   Science (FOCS 2010)},
  pages         = {287--296},
  publisher     = {IEEE},
  year          = {2010},
  doi           = {10.1109/FOCS.2010.34}
}

@article{SS14,
  author        = {Sly, Allan and Sun, Nike},
  title         = {Counting in two-spin models on $d$-regular graphs},
  journal       = {The Annals of Probability},
  volume        = {42},
  number        = {6},
  pages         = {2383--2416},
  year          = {2014},
  doi           = {10.1214/13-AOP888}
}

@article{Valiant79,
  author        = {Valiant, Leslie G.},
  title         = {The complexity of enumeration and reliability problems},
  journal       = {SIAM Journal on Computing},
  volume        = {8},
  number        = {3},
  pages         = {410--421},
  year          = {1979},
  doi           = {10.1137/0208032}
}

@article{vandenBergSteif1994,
  author        = {van den Berg, J. and Steif, Jeffrey E.},
  title         = {Percolation and the hard-core lattice gas model},
  journal       = {Stochastic Processes and their Applications},
  volume        = {49},
  number        = {2},
  pages         = {179--197},
  year          = {1994},
  doi           = {10.1016/0304-4149(94)90132-5},
  url           = {https://ir.cwi.nl/pub/1496/1496D.pdf}
}

@inproceedings{Weitz06,
  author        = {Weitz, Dror},
  title         = {Counting independent sets up to the tree threshold},
  booktitle     = {Proceedings of the 38th Annual ACM Symposium on Theory of Computing (STOC
                   2006)},
  pages         = {140--149},
  publisher     = {ACM},
  year          = {2006},
  doi           = {10.1145/1132516.1132538}
}

\appendix
\section{Regular graphs with independence ratio near one half}
\label{app:tree-extensions}

For completeness, we give the algorithm and proof underlying Corollary \ref{cor:near-half-extension}. The algorithm is almost identical to Algorithm \ref{alg:tree-enumeration}.

\begin{algorithm}[H]
\caption{Approximation by spanning-tree enumeration on regular graphs}
\label{alg:near-half-enumeration}
\begin{algorithmic}[1]
\Require A connected $n$-vertex $\Delta$-regular graph $G$, $\Delta\ge1$,
approximation factor $F>1$ satisfying $2^{-n}\le1-F^{-1}$, and $a\in(0,1/2)$ with
$\alpha(G)\ge(1/2-a/(2\Delta))n$, at fixed fugacity
$\lambda\ge4^{2\Delta/a}$.
\Ensure $F^{-1}Z_G(\lambda)\le\widehat Z\le Z_G(\lambda)$, hence an $F$-factor approximation.
\State Compute a rooted spanning tree $T$ of $G$.
\State Set $K\gets\lfloor an\rfloor$ and $\widehat Z\gets0$.
\For{$b\in\{0,1\}$}
  \ForAll{$S\subseteq E(T)$ with $|S|\le K$}
    \State Set the root occupancy to $b$ and proceed from the root to the leaves.
    \State For each parent--child edge $uv$, set $x_v=1-x_u$ if $uv\notin S$;
    otherwise require $x_u=0$ and set $x_v=0$.
    \State Discard the candidate if a requirement fails.
    \State $I\gets\{v:x_v=1\}$.
    \If{$I$ is independent in $G$}
      \State $\widehat Z\gets\widehat Z+\lambda^{|I|}$.
    \EndIf
  \EndFor
\EndFor
\State \Return $\widehat Z$.
\end{algorithmic}
\end{algorithm}

\begin{proof}[Proof of Corollary~\ref{cor:near-half-extension}]
Put $h(a)=-a\ln a-(1-a)\ln(1-a)$. Choose
$a\in(0,1/2)$ with $h(a)<\ln(1+\varepsilon)$, and set
\[
 \eta=\frac{a}{2\Delta},\qquad \lambda_0=4^{2\Delta/a}.
\]
Set $F=1+\delta$. We may assume $n$ is sufficiently large so that
$2^{-n}\le\delta/(1+\delta)$,
since smaller instances are absorbed by the $\poly(1/\delta)$ factor.
First suppose $G$ is connected. As in
Section~\ref{sec:upper}, call an edge a defect if both endpoints are
unoccupied. The root occupancy and the set of tree defects determine
all occupancies: they alternate across nondefective edges, while both
endpoints of a defective edge must be unoccupied. Thus every retained
independent set is counted exactly once.

Algorithm~\ref{alg:near-half-enumeration} tests at most
\[
 2\sum_{j=0}^{\lfloor an\rfloor}\binom{n-1}{j}
 \le2\sum_{j=0}^{\lfloor an\rfloor}\binom nj
 \le2e^{nh(a)}
\]
candidates, each in polynomial time. This gives running time
$(1+\varepsilon)^n\poly(n,1/\delta)$.

For correctness, every omitted independent set $I$ has more than $an$ tree
defects. By regularity, its total number of defects is
$d(I)=\Delta(n/2-|I|)>an$. Since
$\alpha(G)\ge(1/2-\eta)n$, this gives
\[
 \alpha(G)-|I|>n\left(\frac a\Delta-\eta\right)
 =\frac{an}{2\Delta}.
\]
There are at most $2^n$ omitted sets, and
$Z_G(\lambda)\ge\lambda^{\alpha(G)}$, so
\[
 0\le\frac{Z_G(\lambda)-\widehat Z}{Z_G(\lambda)}
 \le2^n\lambda^{-an/(2\Delta)}
 \le2^{-n}\le\frac{\delta}{1+\delta}.
\]
Consequently, $Z_G(\lambda)/(1+\delta)\le\widehat Z\le Z_G(\lambda)$,
giving the claimed $(1+\delta)$-factor approximation.

For a disconnected graph, let $G_1,\ldots,G_c$ be its components and
put $n_i=|V(G_i)|$. In each $G_i$, use the same encoding to enumerate
independent sets with at most $K=\lfloor an\rfloor$ tree defects,
where $n$ is the order of the whole graph. Let $\widehat Z_i$ be their
total weight, and return $\widehat Z=\prod_{i=1}^c\widehat Z_i$.
The total number of candidates is at most
\[
 2\sum_{i=1}^c\sum_{j=0}^{\min\{K,n_i-1\}}\binom{n_i-1}{j}
 \le2c\sum_{j=0}^K\binom nj
 \le2c e^{nh(a)},
\]
so the running-time bound is unchanged. Every omitted independent
set has more than $K$ tree defects in some component, and hence more
than $an$ defects in total. The preceding tail bound therefore applies
using the independence-ratio promise on $G$.

For sampling, run the enumeration with factor $F=1+\delta/2$ and sample a
retained independent set with probability proportional to its weight.
Writing $\mathcal A$ for the retained sets, the resulting law is
$\nu=\mu_{G,\lambda}(\,\cdot\mid\mathcal A)$, with
\[
 \TV(\nu,\mu_{G,\lambda})=1-\mu_{G,\lambda}(\mathcal A)
 \le\frac{\delta}{2+\delta}\le\delta/2.
\]
This weighted selection can be implemented in a second enumeration
pass with at most $\delta/2$ additional total variation error.
For disconnected $G$, sample independently in each component,
allocating $\delta/(2c)$ additional error to each selection.
This gives total error at most $\delta$ and preserves the running-time bound.
\end{proof}

\section{$\#\mathrm{BIS}$-Hardness on regular bipartite graphs}

\begin{lemma}\label{lem:Delta-reg-bis-hard}
Fix $\Delta\ge3$ and a fugacity
$\lambda>\lambda_c(\Delta)$.
Approximating $Z_G(\lambda)$ on the class $\mathcal G_{\Delta,\Delta}$
is $\BIS$-hard under approximation-preserving reductions.
\end{lemma}

\begin{proof}
We combine Theorem~2 of~\cite{cai} with the degree-completion argument
from Section~\ref{sec:coarse-hardness}. For $r\in(0,1)$ (to be chosen
below), consider the two-spin model with interaction matrix and activity
\[
 B=\begin{pmatrix}r&1\\1&0\end{pmatrix},
 \qquad \lambda_\star=\lambda r^\Delta.
\]
To verify that this model lies in the tree non-uniqueness region, let
$R=Z(1)/Z(0)$ denote the occupation odds at the root of a subtree,
where $Z(i)$ is its partition function with the root fixed to spin $i$.
At a vertex with $\Delta-1$ children, each having occupation odds $R$,
the odds are
\[
f_r(R)=\frac{\lambda_\star}{(r+R)^{\Delta-1}}.
\]
Writing $R=rx$ and using $\lambda_\star=\lambda r^\Delta$, we obtain
\[
\frac{f_r(rx)}r=\frac{\lambda}{(1+x)^{\Delta-1}},
\]
which is the hard-core tree recursion at activity $\lambda$.
Since $\lambda>\lambda_c(\Delta)$, our model also lies in the tree
non-uniqueness region. Theorem~2 of~\cite{cai} therefore implies that
approximating its partition function $Z_G^\star$ on
$\mathcal G^{\leq}_{\Delta,\Delta}$ is $\BIS$-hard.

 For
$G\in\mathcal G^{\leq}_{\Delta,\Delta}$, write $n=|V(G)|$ and
$m=|E(G)|$. Since $B_{11}=0$, its configurations are independent sets,
and their weights give
\[
 Z_G^\star=r^m\sum_{I\in\mathcal I(G)}
 \lambda^{|I|}r^{\sum_{v\in I}(\Delta-d_G(v))}.
\]

To choose $r$ and fill missing degrees, take $Q=K_{\Delta,\Delta}-e$,
with the endpoints of $e$ as ports. Join $k$ copies by connecting each
copy's right port to the next copy's left port. The two free ports lie
on opposite sides and have degree $\Delta-1$ and all other vertices have
degree $\Delta$.

Let $T_{ij}$ be the total hard-core weight in one copy of $Q$ with
port occupancies $i,j\in\{0,1\}$, including their fugacity factors.
The port-weight and hard-core interaction matrices are
\[
 T=\begin{pmatrix}2\tau-1&\lambda\tau\\
                   \lambda\tau&\lambda^2\end{pmatrix},
 \qquad
 M=\begin{pmatrix}1&1\\1&0\end{pmatrix},
 \qquad \tau=(1+\lambda)^{\Delta-1}.
\]
Attach the remaining ports by edges to two external vertices. Including these edges but excluding the external-vertex fugacities,
the resulting gadget contributes $W_k(i,j)$ when the external
occupancies are $i,j$, where
\[
W_k=(MT)^kM.
\]

Let $A=MT$. Since $A$ has strictly positive entries, the
Perron--Frobenius theorem gives a dominant eigenvalue $\rho>0$
and a positive right eigenvector, which we normalize as
$v=(1,r)^{\mathsf T}$. Because the first row of $A$ is strictly larger than the second entrywise, $Av=\rho v$ implies $\rho>\rho r$ and thus $0<r<1$. 

Perron--Frobenius also implies that $\rho^{-k}W_k=\rho^{-k}A^kM$
converges exponentially fast to a positive rank-one matrix whose
columns are proportional to $v$. Since $W_k$ is symmetric, this limit
is a positive multiple of $vv^{\mathsf T}$.
Consequently, setting $C_k=W_k(0,0)$, we obtain
\[
W_k(i,j)=C_k r^{i+j}\bigl(1+O(\theta^k)\bigr),
\]
uniformly for $i,j\in\{0,1\}$, where $\theta\in(0,1)$ is a
constant depending only on $\Delta$ and $\lambda$.

Take the disjoint union of $G$ and a copy with its bipartition
classes exchanged. Each side then has $n$ vertices and
$s=\Delta n-2m$ missing incidences. Pair the missing incidences
across the bipartition and attach one gadget consisting of $k$
copies of $Q$ to each pair. The resulting graph $H$ is bipartite
and $\Delta$-regular.

Up to the approximation error, the gadgets contribute a common
factor $C_k^s$ and one factor $r$ for each missing incidence at an
occupied vertex. Since $s\le\Delta n$, for any $\delta\in(0,1)$
we may choose $k=O(\log((n+1)/\delta))$ so that the accumulated
multiplicative error is at most $e^{\delta/4}$. Summing over the
independent sets in the two copies of $G$ therefore gives
\[
e^{-\delta/4}C_k^s r^{-2m}(Z_G^\star)^2
\le Z_H(\lambda)\le
e^{\delta/4}C_k^s r^{-2m}(Z_G^\star)^2.
\]

The graph $H$ has $2n+2\Delta ks$ vertices, so the construction
has polynomial size. Given a $(1+\delta/2)$-factor approximation
$\widehat Z_H$ to $Z_H(\lambda)$, the estimate
$\widehat Z_G^\star=r^m\sqrt{\widehat Z_H/C_k^s}$ is a
$(1+\delta)$-factor approximation to $Z_G^\star$, completing the
approximation-preserving reduction.
\end{proof}
\end{document}